\documentclass[%
 reprint,
 amsmath,amssymb,
 aps,
]{revtex4-1}
\usepackage{multirow}
\usepackage{amsthm,amssymb}
\usepackage{mathtools}
\usepackage{graphicx}
\usepackage{dcolumn}
\usepackage{bm}
\usepackage{dsfont}
\usepackage{hyperref}
\usepackage{dsfont}
\usepackage[caption=false]{subfig}
\usepackage{float}
\usepackage{afterpage}
\usepackage{mathrsfs}
\usepackage[caption=false]{subfig}
\usepackage{algorithm}
\usepackage{algpseudocode}

\usepackage{bbm}

\newtheorem{theorem}{Theorem}

\newtheorem{lemma}[theorem]{Lemma}

\newcommand{\bbf}{\mathbbm{f}}
\newcommand{\bbF}{\mathbbm{F}}
\newcommand{\bbV}{\mathbbm{V}}
\newcommand{\bbU}{\mathbbm{U}}
\newcommand{\bbi}{\mathbbm{i}}

\newcommand{\ones}{\mathbf{1}}

\newcommand{\Mod}{\operatorname{Mod}}

\newcommand{\Adm}{\operatorname{Adm}}

\begin{document}

\preprint{APS/123-QED}

\title{Generalization of Effective Conductance Centrality for Egonetworks}

\email{Corresponding author: heman@ksu.edu} 
\author{Heman Shakeri$^{1}$}
\author{Behnaz Moradi-Jamei$^{2}$}
\author{Pietro Poggi-Corradini$^3$}
\author{Nathan Albin$^3$}%
\author{Caterina Scoglio$^1$}
\affiliation{$^1$Electrical and Computer Engineering Department, Kansas State University, Manhattan, Kansas, USA}
\affiliation{$^2$Department of Statistics, Kansas State University, Manhattan, Kansas, USA}
\affiliation{$^3$Department of Mathematics, Kansas State University, Manhattan, Kansas, USA}
\date{\today}

\begin{abstract}
	We study the popular centrality measure known as effective conductance or in some circles as information centrality. This is an important notion of centrality for undirected networks, with many applications, e.g., for random walks, electrical resistor networks, epidemic spreading, etc. In this paper, we first
reinterpret this measure in terms of modulus (energy) of families of walks on the network. This modulus centrality measure coincides with the effective conductance measure on simple undirected networks, and extends it to much more general situations, e.g., directed networks as well. Secondly, we study a variation of this modulus approach in the egocentric network paradigm. Egonetworks are networks formed around a focal node (ego) with a specific order of neighborhoods.  We propose efficient analytical and approximate methods for computing these measures on both undirected and directed networks. Finally, we describe a simple method inspired by the modulus point-of-view, called \textit{shell degree}, which proved to be a useful tool for network science.
\end{abstract}

\maketitle
The concept of information centrality was first introduced in \cite{stephenson1989rethinking} and was later reinterpreted in terms of electrical conductance in \cite{klein1993resistance}. Given a network $G=(V,E)$ and a node $a\in V$, the effective conductance centrality of $a$ is defined as
\begin{equation}\label{eq:effcondcentr}
\mathcal{C}_{\text{eff}}(a):=\sum_{b\in V\setminus a}\frac{1}{\mathcal{R}_{\text{eff}}(a,b)}.
\end{equation}
where  $\mathcal{R}_{\text{eff}}(a,b)$ is effective resistance distance between $a$ and $b$.
Note that this measure considers every possible path that electrical current flow might take from $a$ to an arbitrary sink $b$.

The situation can be clarified by introducing the notion of  modulus of families of walks. This is a way of measuring the richness of certain families of walks on a network (and beyond, see \cite{shakeri2017complex, shakeri2017loop, shakeri2016generalized}). 
Given two nodes $a$ and $b$ we may consider the connecting family $\Gamma(a,b)$ of all walks $\gamma$ from $a$ to $b$. Then, given edge density $\rho:E\rightarrow \mathbb{R}$ for $p\in[1,\infty]$, we define 
$\ell_\rho\left(\Gamma\right):=\min_{\gamma\in\Gamma}\ell_\rho\left(\gamma\right)$ where $\ell_\rho(\gamma)$ is the {\it $\rho$-length} of a walk $\gamma$: 
\begin{equation}
\ell_\rho\left(\gamma\right):=\sum_{e\in \gamma}\rho\left(e\right).
\end{equation}
The $p$-modulus of $\Gamma$ is defined as 
\begin{equation}\label{eq:modulus}
\Mod_{p}\left(\Gamma\right) := \min_{\ell_\rho(\Gamma) \geq 1}\text{Energy}_{p}\left(\rho\right)
\end{equation}
Namely, we minimize the energy of candidate edge-densities $\rho$ subject to the $\rho$-length of every walk in $\Gamma$ being greater than or equal one, \textit{i.e.}, $\ell_\rho(\Gamma)\ge 1$.
These densities can be interpreted as costs of using the given edge. The energy we consider is 
\begin{equation}\label{eq_energy_p}
\text{Energy}_{p}(\rho) =\sum_{e\in E} \sigma(e)\rho\left(e\right)^p,
\end{equation}
where $\sigma(e)> 0$, is the conductance of the edge $e$.
Thus modulus is a constrained convex optimization problem that has a unique extremal density $\rho^*$ when $1<p<\infty$.
This point of view allows for much more flexibility, because it can be applied to a variety of different families of objects: walks, cycles, trees, etc, and also works when the underlying network is directed or weighted.
Moreover, modulus has very useful  properties of $\Gamma$-monotonicity and countable subadditivity. 

	
	For undirected networks the effective conductance between $a$ and $b$ is connected to  $\Mod_2(\Gamma(a, b))$ as follow \cite{duffin:1962jmaa, Albin2014} 
	\begin{equation}\label{eq:Cmod}
	\frac{1}{\mathcal{R}_{\text{eff}}(a,b)} = \Mod_2(\Gamma(a, b)).
	\end{equation}
	 In the following, we reproduce a proof for this connection and how to calculate $\Mod_2(\Gamma(a, b))$ in symmetric networks using the pseudoinverse of the Laplacian.

Let $\bbF$ be the set of all unit flows $\bbf:E\rightarrow \mathbb{R}$ that satisfy Kirchoff’s node law and pass through a network $G$ from $a$ to $b$. Namely for $v\in V$
\[
(\nabla.\bbf)(v) =  \begin{cases} 
1 & v= a \\
-1 & v=b \\
0 & \text{otherwise}
\end{cases}
\]
corresponds to the injected currents at each node.
The energy of $\bbf$ is
\[
\text{Energy}(\bbf) := \sum_{e\in E} \mathcal{R}(e)\bbf(e)^2
\]
where $\mathcal{R}(e) = \frac{1}{w(e)}$ is the resistance of edge $e$.
A unit current flow $\bbi\in \bbF$ is a unit flow that also satisfies Ohm's law, \textit{i.e.}, there is a function 
$\bbV :V\rightarrow \mathbb{R} $ (called a potential) such that
for every edge $(a,b)$:
\[
\mathcal{R}(a,b)\bbi(a,b) = \bbV(b) - \bbV(a).
\]
%

Let $\bbU:V \rightarrow \mathbb{R}$ be a vertex potential function. We can define a density $\rho_\bbU$  as the gradient of $\bbU$, \textit{i.e.}, for the edge $e=\lbrace v,w\rbrace$
\[
\rho_\bbU(e) := | \bbU_u - \bbU_w | 
\]
Then, $\rho_\bbU$ is admissible  for walks from $a$ to $b$, whenever $\bbU(a) = 0,~\bbU(b) = 1$.

Conversely, if $\rho$ is an admissible density, then we can define a potential $\bbU(x)$ as the infimum of $\ell_\rho(\gamma)$ over all walks from $a$ to $x$. With this definition, $\rho_\bbU=\rho$, see \cite{Albin2014}.

In particular, assuming each edge has a unit resistance,
\[
\text{Energy}(\rho_\bbU) =\sum_{e\in E} \rho_\bbU(e)^2=\sum_{e=\{u,w\}\in E}|\bbU(u)-\bbU(w)|.
\]
Hence,  if we substitute $\bbU$ with $\frac{\bbV}{\mathcal{R}_{\text{eff}}(a,b) }+C$, 
where $\bbV$ is the electric potential when a unit current flow $\bbi\in \bbF$ is passing through the network with source $a$ and sink $b$ and the effective resistance between $a$ and $b$ is $\mathcal{R}_{\text{eff}}$, then, 
\begin{equation}\label{eq_mod_Ceff}
\begin{split}
\Mod_2(a,b) = \min_{\substack{\bbU_a =0\\ \bbU_b = 1}} \rho_\bbU^T\rho_\bbU  = \frac{1}{\mathcal{R}_{\text{eff}}(a,b)}.
\end{split}
\end{equation}

By Kirchhoff's law of current conservation:
\[
\sum_j a_{i,j} (\bbV_{i} -  \mathbbm{V}_j) =  (\nabla.\mathbbm{i})(i)
\]
where $A=[a_{ij}]\in \mathbb{R}^{N\times N}$ is the adjacency matrix of $G$, with $a_{ij}=1$ if and only if ${i,j}\in E$. In matrix form:
\begin{equation}\label{eq_laplacian}
L \bbV  = \mathbbm{I}
\end{equation}
where $L$ is the Laplacian matrix of $G$ and $\mathbbm{I} =\nabla.\mathbbm{i}$. Because $\bbV$ is defined up to an additive and the nullspace of
$L$ is along the constant vector, we ground an arbitrary node $k$ and thus reduce $L$ by removing $k$th row and column denoted by $\prescript{k}{}L$ \cite{van2017pseudoinverse}. Now we can find solve \eqref{eq_laplacian}:
\[
\prescript{k}{}{\bbV} = (\prescript{k}{}{L})^{-1}~ \prescript{k}{}{\mathbbm{I}}.
\]
we denote $(\prescript{k}{}{L})^{-1}$ by $\mathcal{G}$ (reduced conductance matrix)
and obtain effective resistance between nodes $a$ and $b$ is
\begin{equation}
\begin{split}
\mathcal{R}_{\text{eff}}(a,b) &=\prescript{k}{}{\bbV_a} - \prescript{k}{}{\bbV_b}\\
&= \mathcal{G} _{a,a}+\mathcal{G} _{b,b}-2~\mathcal{G} _{a,b} 
\end{split}
\end{equation}
and from \eqref{eq_mod_Ceff}:
\begin{equation}
\Mod_2(a,b) = \left( \mathcal{G}_{a,a}+\mathcal{G}_{b,b}-2\mathcal{G}_{a,b}  \right)^{-1}
\end{equation}

Therefore, using  \eqref{eq:Cmod}, we can rewrite the effective conductance centrality in \eqref{eq:effcondcentr} in the Modulus language
\begin{equation}\label{eq_CFC}
\mathcal{C}_{\text{eff}}(a)=  \sum_{b \in V\setminus a} \Mod_2(\Gamma(a, b)).
\end{equation}

 For the rest of this paper, we consider $p=2$ due to its physical interpretation as effective conductance as well as computational advantages, for instance, in this case \eqref{eq:modulus} is a quadratic program. Moreover, the right-hand side also makes sense on directed networks.

\section{Egocentric effective conductance centrality}
As mentioned above, $\mathcal{C}_{\text{eff}}(a)$ is sociocentric in the sense that it considers all walks from $a$ to an arbitrary node in $G$. However, in practice, it can be prohibitive to scale sociocentric methods to very large networks. Moreover, in real-world situations it is not feasible to have access to the entire network. Rather,  one can at best know local information up to a few neighborhood levels. For instance, when data is anonymized to protect privacy of network entities, identifying the sociocentric picture is impossible, \textit{e.g.}, sexual networks may be limited to the number of contacts of individuals.

An alternative approach is to consider measures that are adapted to egonetworks (also known as neighborhood networks).
An ego network $G^a(r)$ around a node $a$ is constructed by collecting data (nodes and edges) starting from the ego $a$ and searching $G$ out to a predefined order of neighborhood $r\in \lbrace 1,\cdots \epsilon(a)\rbrace$; where $\epsilon(a)$ is the eccentricity of node $a$ or the maximum distance from $a$ to nodes in $G$. 

Egonetworks are often preferred because they support more flexible data collection methods \cite{carrasco2008collecting} and often involve less expensive computation costs. Egocentric measures are more stable \cite{costenbader2003stability} against network sampling and reliable (less sensitivity) with measurement errors \cite{zemljivc2005reliability}.
We concentrate on unweighted (binary) networks to simplify the algebra,
although, all of our methods and discussions can be easily generalized for weighted networks. Thus, we let $d(a,b)$ denote the shortest-path distance between two nodes (smallest number of hops). The neighborhood structure around an ego $a$ is described by the shells of order $k$:
\[
S(a,k):=\lbrace y\in V: d(a,y)=k\rbrace,
\]
and the corresponding families of walks
$\Gamma (v,S(a,k))$, consisting of simple walks that begin at ego $v\in V$ and reach  $S(a,k)$ for the first time.
Modulus allows a quantification of the richness of the family of walks, \textit{i.e.}, a family with many short walks has a larger modulus than a family with fewer and longer walks.  Here we consider {\it shell modulus} $\Mod_2(v,S(a,k))$ which quantifies the capacity of walks emanating from the ego up to the shell $S(a,k)$ \cite{shakeri2016generalized} without having to account the data outside $G^a(k)$.
\begin{theorem}\label{thm_ShellAnal}
	For undirected networks, we can calculate $2$-modulus of $\Gamma \left(v,S(a,k)\right)$ analytically without going through the optimization problem in \eqref{eq:modulus}:
	\begin{equation}\label{eq_anal_ego}
	\Mod_2(a, S(a,r)) = \frac{1+x_s\sum_{j =S_1}^{S_s-1} \frac{1}{x_i}}{x_s} 
	\end{equation}
	where $x_i= \sum_{j =S_1}^{S_s-1}  \mathcal{G}_{ij } $ .
\end{theorem}

\begin{proof}
	Similar to \eqref{eq:Cmod}, to find $\Mod(a, S(a,r))$ in $G^a(r)$, we solve Kirchhoff's law of currents 
	\begin{equation}\label{eq_LapEgo}
	L^a_{(r)}\bbV = \mathbbm{I}
	\end{equation}
	where $L^v_{(r)}$ is the Laplacian matrix of $G^a(r)$ and $\mathbbm{I}$ is the applied external current vector with values $1$ at ego and for nodes in $S(a,r)$
	\begin{equation}\label{eq_Is}
	\ones^T \mathbbm{I}_S = -1
	\end{equation}
	and zero for other nodes (see Figure \ref{fig_Elec_Shell}).
	\begin{figure}
		\includegraphics[clip,width=.8\columnwidth]{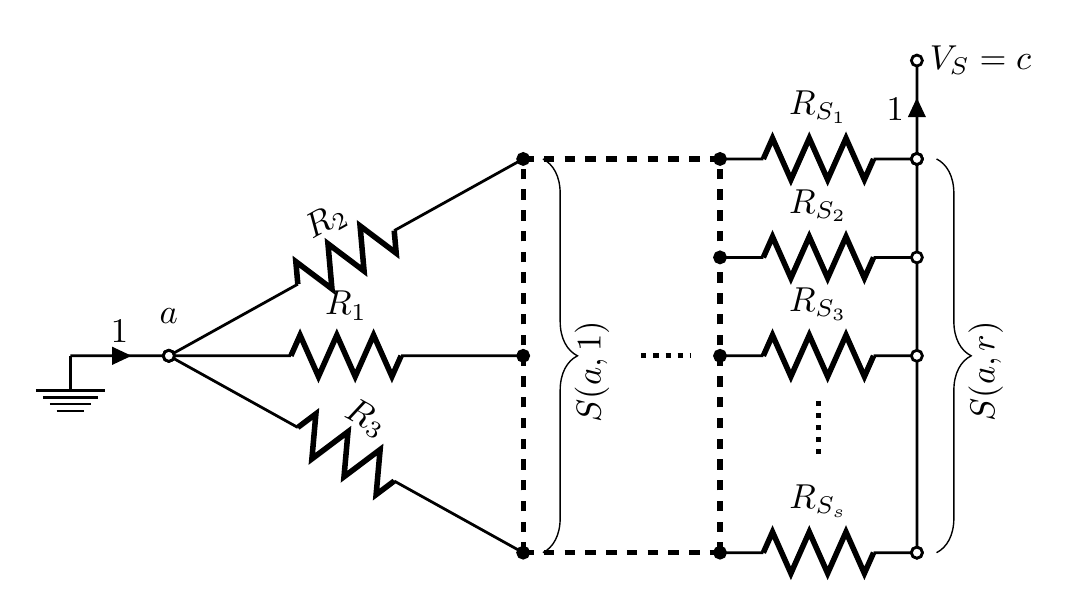}%
		\caption{Interpreting $\Mod_2(a,S(a,r))$ as finding effective conductance between grounded node $a$ and nodes with the same potential $c$ in $S(a,r)$ in an electrical network. Solution follows from the corresponding Laplacian system.}\label{fig_Elec_Shell}
	\end{figure}
	Nodes in $S(a,r)$ have similar electric potential $c$.
	
	The above problem has a unique harmonic solution for $\bbV$ up to a constant, we ground the potential at ego, \textit{i.e.}, $\bbV_a = 0$ and find other nodes potentials by
	\[
	\bbV = \mathcal{G}~\mathbbm{I}
	\]
	where $\mathcal{G} = \left(\prescript{a}{}{L^v_{(r)}}\right)^{-1}$ is the reduced conductance matrix. Combining \eqref{eq_LapEgo} and \eqref{eq_Is}
	\begin{align}\label{eq_matrixShell}
	\begin{bmatrix}
	V_2\\
	\vdots\\
	c\\
	c\\
	\vdots\\
	c
	\end{bmatrix} =\mathcal{G}
	\begin{bmatrix}
	0\\
	\vdots\\
	\mathbbm{I}_{S_1}\\
	\mathbbm{I}_{S_2}\\
	\vdots\\
	\mathbbm{I}_{S_s}
	\end{bmatrix}
	\rightarrow
	\begin{bmatrix}
	V_2\\
	\vdots\\
	\frac{c}{\mathbbm{I}_{S_1}}\\
	\frac{c}{\mathbbm{I}_{S_2}}\\
	\vdots\\
	\frac{c}{-1-\sum_{j =S_1}^{S_s-1} \mathbbm{I}_{j} }
	\end{bmatrix} =\mathcal{G}
	\begin{bmatrix}
	0\\
	\vdots\\
	1\\
	1\\
	\vdots\\
	1
	\end{bmatrix}
	=\mathbf{x}
	\end{align}
	where $x_i= \sum_{j =S_1}^{S_s-1}  \mathcal{G}_{ij } $.
	If $|S|=s$ and for $i\in \{S_1, \cdots,S_{s-1}\}$
	\[
	\mathbbm{I}_i = \frac{c}{x_i}
	\]
	From \eqref{eq_matrixShell}:
	\[
	\frac{c}{-1-c\sum_{j =S_1}^{S_s-1} \frac{1}{x_i}  }  =x_{S_s}
	\]
	\[
	c = \frac{-x_s}{1+x_s\sum_{j =S_1}^{S_s-1}\frac{1}{x_i}}
	\]
	and the effective resistance between $a$ and $S(a,r)$:
	\[
	\mathcal{R}_{a, S(a,r)} = \bbV_v - c = \frac{x_s}{1+x_s\sum_{j =S_1}^{S_s-1} \frac{1}{x_i}}
	\] 
	and since $\bbV_a = 0$ (grounded):
	\[\Mod_2(a, S(a,r)) = \frac{1+x_s\sum_{j =S_1}^{S_s-1} \frac{1}{x_i}}{x_s}.
	\]
	
\end{proof}

The convex optimization problem in \eqref{eq:modulus} involves a quadratic minimization. In the undirected case, computing the pseudoinverse of the Laplacian in \eqref{eq_anal_ego} involves solving a Laplacian system. In both cases, algorithms and technique are still improving and advancing. However, graphs with more than a million edges may become untractable.

We propose the following egocentric version of $\mathcal{C}_{\text{eff}}(a)$ using shell modulus:
\begin{equation}\label{eq_ShellCFC}
\mathcal{C}_\text{shell}(a, r) := \sum_{k=1}^{r} \Mod_2(v, S(a,k)) 
\end{equation}
This {\it shell modulus centrality}  follows the same logic as \eqref{eq_CFC} but only requires the egocentric network data. For undirected networks, we can analytically compute \eqref{eq_ShellCFC} using Theorem \ref{thm_ShellAnal}.

In Figure \ref{fig_ShellC}, centralities of nodes in three small networks are computed, by considering $\mathcal{C}_{\text{shell}}(v,r)$ with $r={\rm diam}(G)$. 
\begin{figure}
	\subfloat[]{%
		\includegraphics[clip,width=.35\columnwidth]{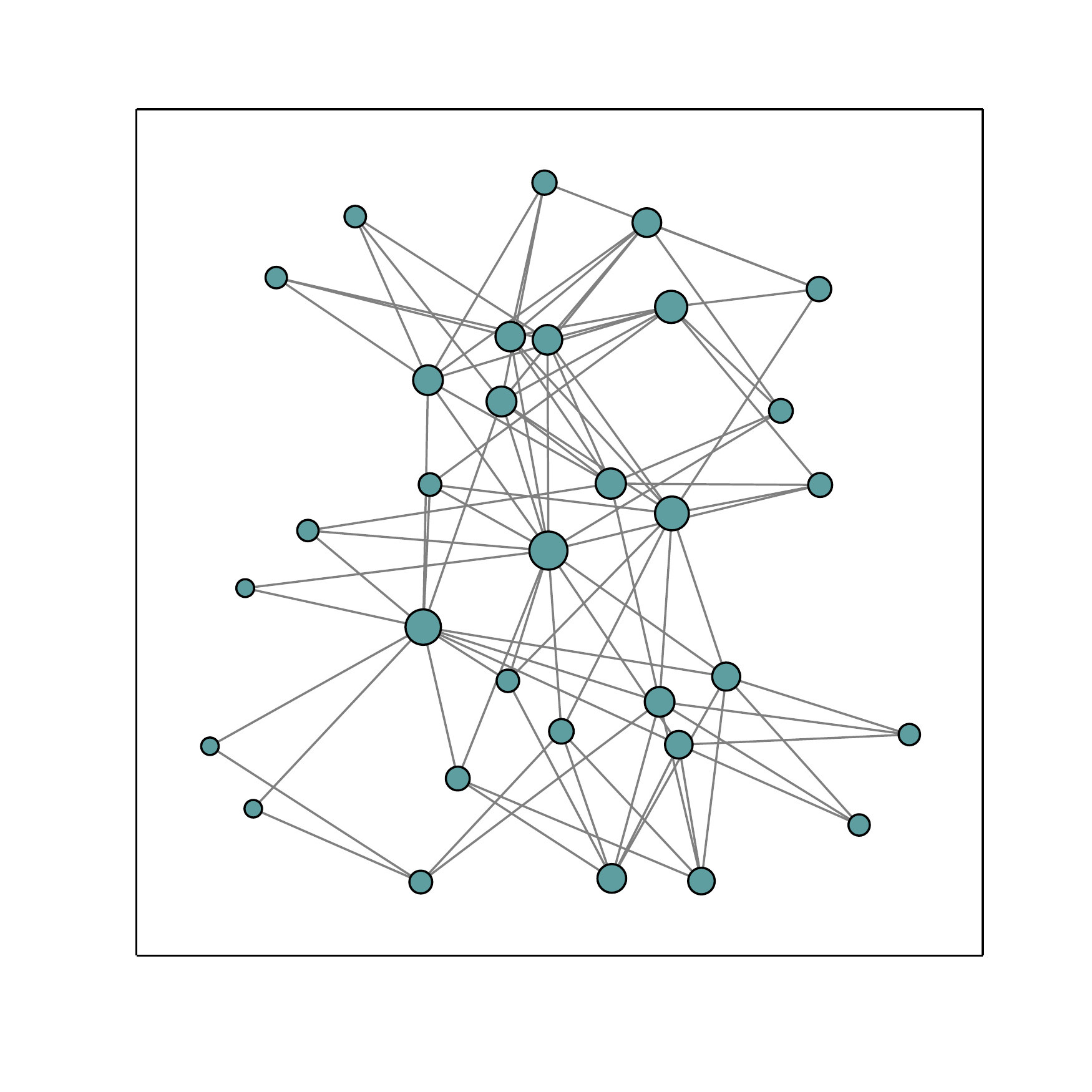}%
z	}
	\subfloat[]{%
		\includegraphics[clip,width=.35\columnwidth]{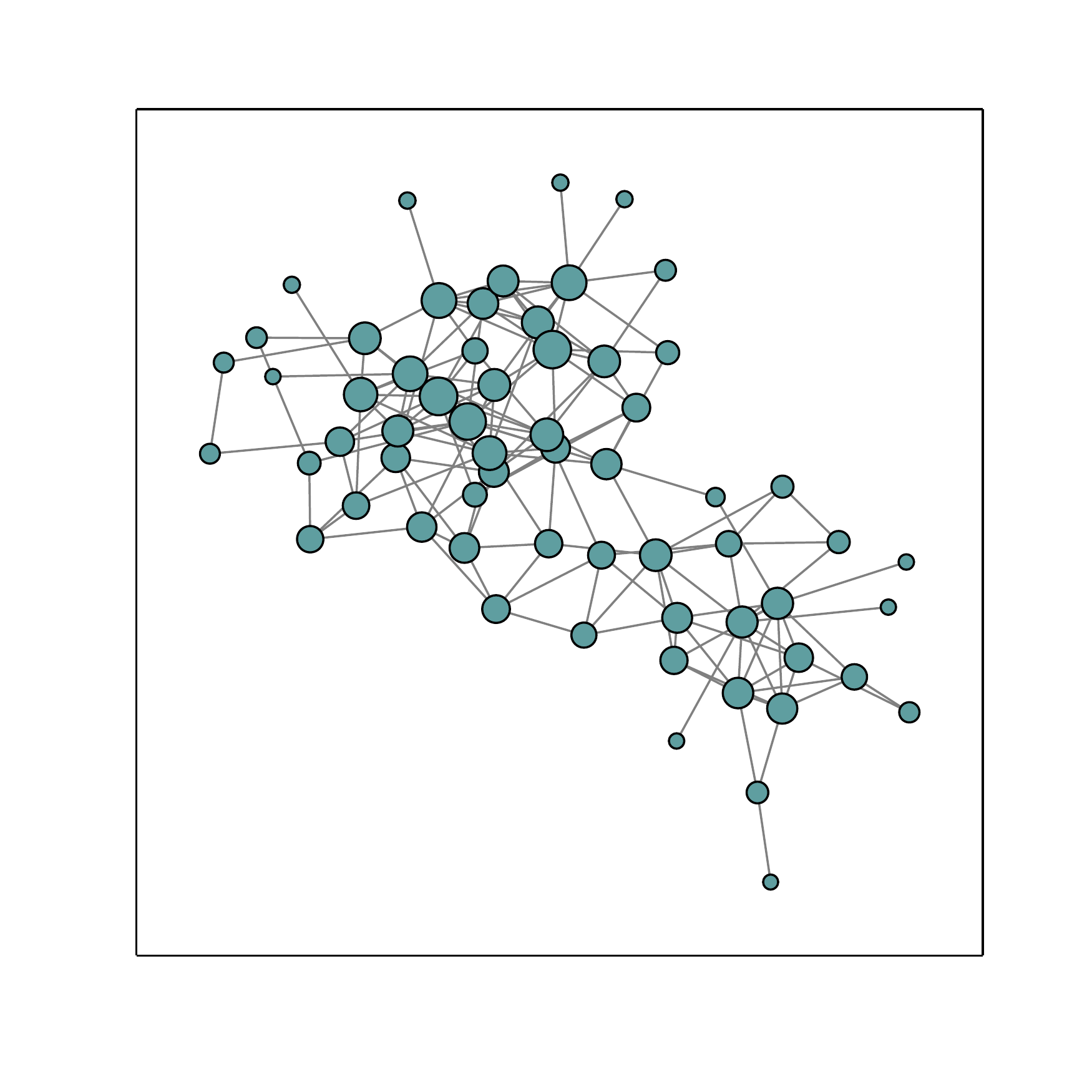}%
	}
	\subfloat[]{%
		\includegraphics[clip,width=.35\columnwidth]{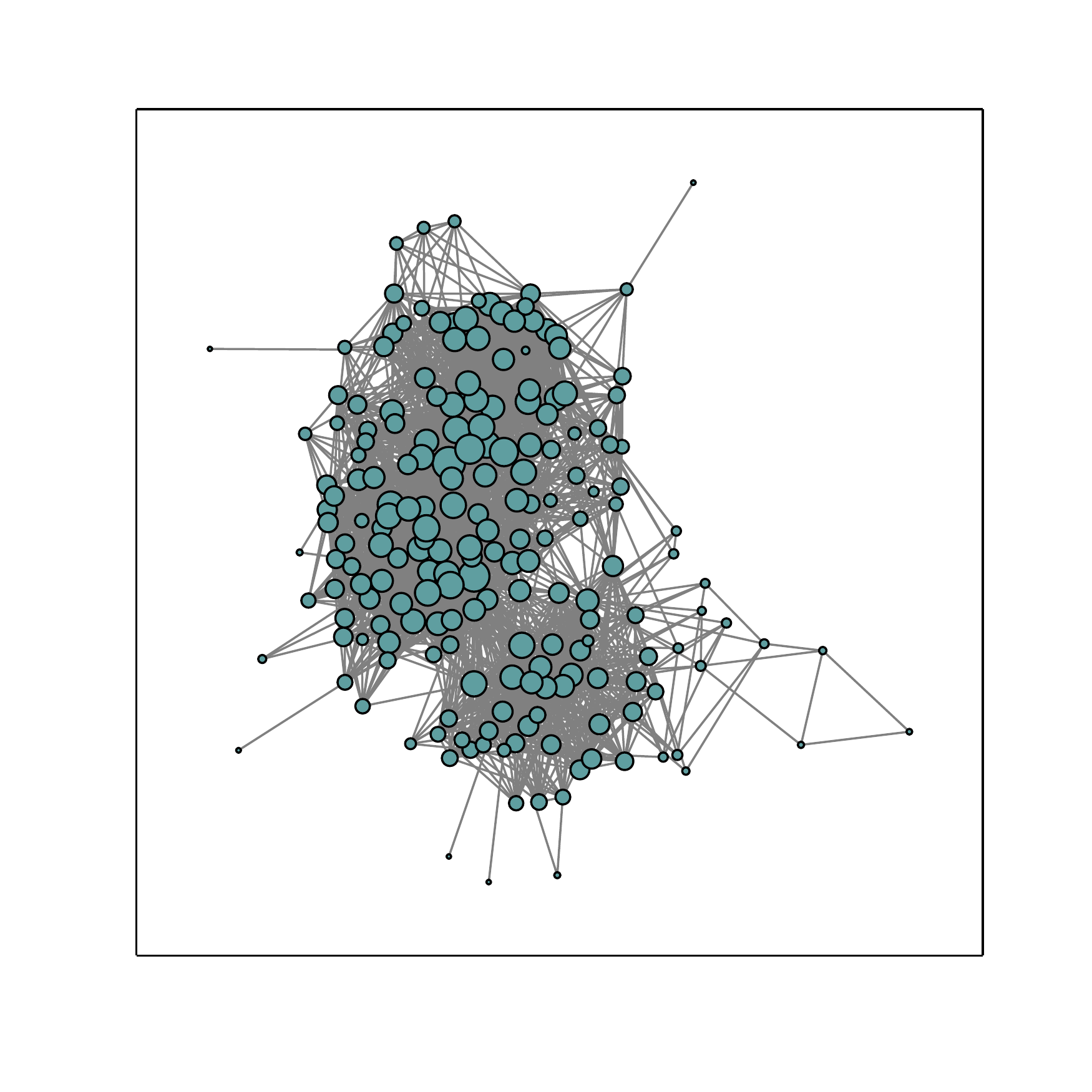}%
	}
	\caption{(a) Davis southern women social network \cite{davis2009deep}. (b) Social network of bottlenose dolphins \cite{lusseau2003bottlenose}. (c) Jazz musicians network  \cite{Jazz}. Node sizes are scaled with the egocentric version of effective conductance centrality computed by \eqref{eq_ShellCFC}. The ranking is unchanged when using  the sociocentric version \eqref{eq:effcondcentr}}\label{fig_ShellC}
\end{figure}
In Figure \ref{fig_ShellC}(a-c), node sizes are scaled with their $C_\text{shell}(v, r)$ values and the computed centralities give, as expected, the same ranking as effective conductance.

In general, $\eqref{eq_CFC}$ requires $|V |$ modulus computations in all of $G$, while $\eqref{eq_ShellCFC}$ only needs $r$ modulus computations in $G^a(r)$. 

Shell modulus centrality  can handle fairly large networks, \textit{e.g.} 100,000 edges. The algorithm used here computes $\eqref{eq:modulus}$ using an active set dual method quadratic programming \cite{goldfarb1983numerically}. It's theoretically enough to consider at most $|E|$ active constraints \cite{albin2016minimal}. Violated (active) constraints are found using Dijkstra's algorithm and the constraint matrix is updated using the  Cholesky decomposition. 

 In the following, we focus on approximating \eqref{eq_ShellCFC} efficiently, while incorporating most of the benefits of shell modulus in a scalable framework.

\subsection{Bounding from above}\label{sec_upper}

First, we provide an upper bound that is known in the complex analysis literature as {\it Ahlfors estimate} \cite[Chapter 4, Equations 4-6]{AhlforsBook}, and in the context of electrical networks goes under the name of Nash-Williams inequality \cite{lyons2016probability}.
Given an egonetwork $G^a(r)$, we consider  the set of edges that connect a shell $S(a,k-1)$ to the next shell $S(a, k)$, for $k\in \lbrace 1,\cdots,r\rbrace$: 
\[
E(a,k):= \left\lbrace e=\{x,y\}\in E|\ x\in S(a,k-1),\ y\in S(a,k)\right\rbrace.
\]
We call the sets $E(a,k)$ {\it shell connecting sets}.
Since  $\Mod_2(v,S(a,r))$ is a minimization problem (\ref{eq:modulus}), we get an upper bound simply by choosing an appropriate admissible density $\bar{\rho}$. Here, we pick the best admissible density that is constant for all edges in each shell connecting set. After computing the minimized energy of this density, we obtain the following upperbound:
\begin{theorem}[Ahlfors upperbound]
	Shell modulus is bounded by the following inequality
\begin{equation}\label{eq:upperbound}
\Mod_2(a,S(a,r))\le  \frac{1}{\sum_{k=1}^{r}\frac{1}{|E(a,k)|}}.
\end{equation}
\end{theorem}

\begin{proof}
	Since (\ref{eq:modulus}) is a minimization problem, an upper bound for the shell modulus $\Mod_2(a,S(a,r))$ can be found by picking an appropriate density  $\bar{\rho}$. Here we will restrict ourselves to densities that are constant on the shell connecting sets $E(a,k)$. Let
	\[
	\bar{\rho}(e):=x_k\qquad\text{if $e\in E(a,k)$}.
	\]  
	Then we solve the following minimization problem:
	\begin{equation}\label{eq_min_xk}
	\begin{aligned}
	& \underset{x}{\text{minimize}}
	& & \sum_{k=1}^r\theta_k x_k^2 \\
	& \text{subject to}
	& & \sum_{k=1}^r x_k =1
	\end{aligned}
	\end{equation}
	where $\theta_k:=|E(a,k)|$. By Cauchy-Schwarz inequality
	\[
	1\le 
	\left( \sum_{k=1}^r x_k\right)^2=\left( \sum_{k=1}^r \frac{1}{\sqrt{\theta_k}}\sqrt{\theta_k}x_k\right)^2 \le \sum_{k=1}^r \frac{1}{\theta_k}\sum_{k=1}^r\theta_k x_k^2
	\]
	and thus the minimum in \ref{eq_min_xk}  is greater than $\left(\sum_{k=1}^r \frac{1}{\theta_k}\right)^{-1}$. However,  when $x$ takes the form:
	\[ x_k =\frac{C}{\theta_k},
	\]
	the minimum is achieved for
	\[
	C = \frac{1}{\sum_{k=1}^{r} \frac{1}{\theta_k}}.
	\]
	
\end{proof}

\subsection{Bounding from below}\label{sec_lower}
To provide a lower bound for shell modulus, we focus on
geodesic paths (shortest walks). These are usually the most important pathways of influence between the ego and other nodes. Classical measures of centrality, such as closeness centrality and betweenness centrality, are  based uniquely on shortest paths \cite{freeman1978centrality}.  

When collecting the egocentric data around an ego $a$, one can take care to avoid forming cycles, and the  resulting egonetwork becomes a tree. So assuming $T^a(r)$ is a tree contained in $G^a(r)$, we can use $\Gamma$-monotonicity to get a lower bound, i.e., if $\Gamma'\subset\Gamma$, then $\Mod_2\left(\Gamma'\right)\leq \Mod_2\left(\Gamma\right)$ \cite{shakeri2016generalized}.

 Moreover, if we write $\Mod_2(T^a(r))$ for the shell modulus 
of all walks in $T^a(r)$ starting at the root $a$ and reaching depth-level $r$, 
this can be analytically calculated.
\begin{theorem}
$T^a(r)$ can be calculated using the following
 recursive formula.
\begin{equation}\label{eq:recursion}
\Mod_2(T^a(k)) = \sum_{c\in C(a)}\frac{\Mod_2(T_{c,k-1})}{1+\Mod_2(T_{c,k-1})}
\end{equation}
where $C(a) := \{c_1 , c_2 , . . . , c_m \} \subseteq V$ are the children of $a$ and $T_{c,k-1}$ represents the subtree 
formed from $T_a$ by keeping only $c$ and its descendants.
\end{theorem}
To prove Equation \eqref{eq:recursion}, let $T_a$ be a rooted shortest tree at $a$ with vertex set $V$, and edge set $E$. Every density $\rho : E \rightarrow [0, \infty)$
gives a weighted distance on the tree defined by
\[
d_{\rho} (x, y) = \sum _{e\in \gamma(x,y)}\rho(e)
\]
We define the set of admissible densities Adm$(T^a_{k}  )$, for walks starting from root $a$ (ego) to leaves at depth $k$, denoted by $l_k$
\[
\text{Adm}(T^a_{k} ) := \lbrace\rho : E \rightarrow [0, \infty) : \ell_\rho(a, l_k)\ge 1 \rbrace.
\]
with modulus
\[
\Mod_2 (T^a_{k}  ) :=\inf_{\rho\in \text{Adm}(T_{a,k})} \sum_{e\in E} \rho(e)^2
\]

\begin{figure}
	\includegraphics[clip,width=.8\columnwidth]{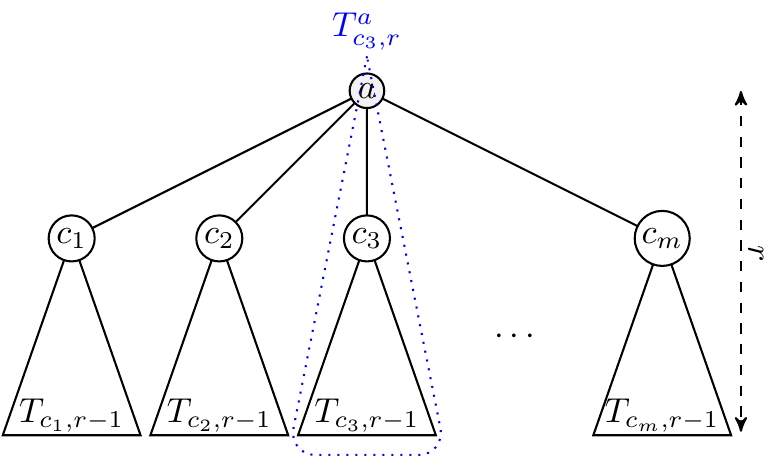}%
	\caption{The tree $T^a_{r}$ and its subtrees. Each child $c_i$ of $a$ can induce two subtrees--if it has descendants until depth $r-1$. $T_{c_i,r}^a$
		(outlined
		with a dashed line for $i = 3$ in the figure) is the subtree rooted at $a$ formed by removing all other
		children and their descendants from $T^a_r$. $T_{c_i,r-1}$
		is the subtree rooted at $c_i$ formed by removing $a$ from
		$T_{c_i,r}^a$
		.}\label{fig_Tree}
\end{figure}

Assuming $a$ has at least one child, let $C(a) := \{c_1 , c_2 , . . .  \} \subseteq V$ be the children. Each child
$c$ induces two rooted subtrees (Figure \ref{fig_Tree}). Let $T_c^a$ represent the subtree (still rooted at $a$)
formed from $T_v$ by pruning all of $a$’s children other than $c$ along with their descendants, and let $T_c$
represent the subtree (now rooted at $c$) formed by removing $a$ from $T_c^a$.

The following lemma is
an immediate consequence of the \emph{parallel rule} of modulus: Given two families $\Gamma_1$ and $\Gamma_2$, suppose that $e\in E$ and $\gamma_1\in \Gamma_1$ and $\gamma_2\in \Gamma_2$ we have $\mathcal{N}(\gamma_1,e)\mathcal{N}(\gamma_2,e)=0$, then
$\Mod_2\left(\Gamma_1\cup\Gamma_2\right) = \Mod_2\left(\Gamma_1\right)+\Mod_2\left(\Gamma_2\right)$.

\begin{lemma}\label{lem_TC1}
	The modulus of $T^a_{k} $ is related to the moduli of the $T_{c_i, k}^a$ as follows.
	\[
	\Mod_2 (T^a_{k}  ) = \sum_{i=1}^{m}\Mod_2 (T_{c_i,k}^a ).
	\]
\end{lemma}

By Lemma \ref{lem_TC1}, we may restrict ourselves to the case that $a$ has a single child $c$. In this case, the \emph{serial rule} for modulus allows us to reduce the problem to finding the modulus of $T_{c,k-1}$. This is explained in the following lemma.
\begin{lemma}\label{lem_TC2}
	The modulus of $T_{c,k}^a$ is related to the modulus of $T_{c,k-1}$ as follows.
	\begin{equation}\label{eq_0Recursion}
	\Mod_2(T_{c,k}^a) = \frac{\Mod_2(T_{c,k-1})}{1+\Mod_2(T_{c,k-1})}
	\end{equation}
	
\end{lemma}

\begin{proof}
	If $c$ is a leaf of $T^a_{k}$ , then $\rho(a, c) = 1$ is the minimizer for the modulus. Otherwise, by
	considering the density, $\rho(v, c)$, on the edge from $a$ to $c$, the optimization effectively decouples. In
	order for $\rho$ to be admissible, it is necessary that $d_\rho (c, l) \ge 1- \rho(a, c)$ for every leaf $l_{k-1}$ of $T_{c,k-1}$ at depth $k-1$.
	For $0 \le \ell \le 1$, define the parameterized set of admissible densities, for every leaf $l_{k-1}$
	\[
	\text{Adm} (T_{c,k-1} ;\ell ) := \lbrace \rho : E \rightarrow [0, \infty) : d(c, l_{k-1}) \le \ell \rbrace
	\]
	and the parameterized modulus problem
	\[
	\Mod_p'(T_{c,k-1} ;\ell ) =\inf_{\rho\in \text{Adm}'(T_{c,k-1};\ell)}\sum_{e\in E(T_c)} \rho(e)^2
	\]
	
	where $E(T_{c,k-1} )$ represents the set of edges in the subtree $T_{c,k-1}$ . It is straightforward to verify that
	\[
	\Mod_2'(T_{c,k-1};\ell) = \ell^2 \Mod_2(T_c)
	\]
	and, thus
	\begin{equation}
	\begin{split}
	\Mod_2(T^a_{c,k}) =& \inf_{0\le \rho(v,c) \le 1 }\lbrace \rho(a,c)^2 +\\
	& \Mod_2'(T_{c,k-1}:1-\rho(v,c)) \rbrace\\
	 =& \inf_{0\le \rho(a,c) \le 1 } \lbrace \rho(a,c)^2 +\\
	&(1- \rho(k,c))^2 \Mod_2(T_{c,k-1}) \rbrace
	\end{split}
	\end{equation}
	The infimum, given by \eqref{eq_0Recursion}, is attained when
	\[
	\rho(a,c) = \frac{\Mod_2(T_{c,k-1})}{1+\Mod_2(T_{c,k-1})}
	\]
\end{proof}

Lemmas \ref{lem_TC1} and \ref{lem_TC2} combined prove Equation \eqref{eq:recursion}.
%
%


Equation \eqref{eq:recursion} computes $\Mod_2(T^a(k))$ recursively. For each leaf
node $l_k$, set $\Mod_2(T_{l_k, 0}) = \infty$. Then \eqref{eq:recursion} will propagate the modulus to the ego. For example, to compute $\Mod_2(T_{\text{a},2})$ in the graph in Figure \ref{fig:S2L}(b), we start by assigning $\infty$ for modulus of the leaves $e$ and $f$. Then, by \eqref{eq:recursion}, each contributes $1$ to node $b$, and $\Mod_2(T_{\text{b},1}) =2$. Thus $\Mod_2(T^{\text{a}}(2)) =\frac{\Mod_2(T_{\text{b},1}) }{1+\Mod_2(T_{b,1}) } = \frac{2}{3}$.


\section{Shell Degree}\label{sec:gendeg}
In conclusion, Ahlfors' upper bound \eqref{eq:upperbound} considers all edges in the shell connecting sets even if they are not on the shortest paths, such as edge $a-d$ in Figure \ref{fig:S2L}(a). 
On the other hand, when using the ego-tree approximation, we inevitably lose valuable information hidden in the edges that where discarded.
For example, in Figure \ref{fig:S2L}(b-c), to form a tree we need to solve the child custody problem between parents $b$ and $c$ and child $f$. In particular, the lower bound calculation will discard at least one edge. Moreover, this leads to multiple possible lower bounds, \textit{e.g.}, $\Mod_2(T_{a,r}) =\frac{2}{3} $ in Figure \ref{fig:S2L}(b) and $\Mod_2(T_{a,r}) = 1$ for Figure \ref{fig:S2L}(d). 

\begin{figure}
	\subfloat[]{%
	\includegraphics[clip,width=.35\columnwidth]{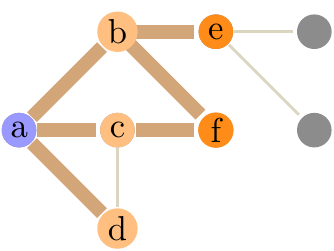}%
}~~~
	\subfloat[]{%
		\includegraphics[clip,width=.35\columnwidth]{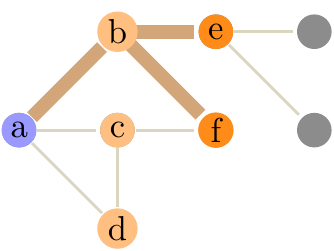}%
	}\\
	\subfloat[]{%
		\includegraphics[clip,width=.35\columnwidth]{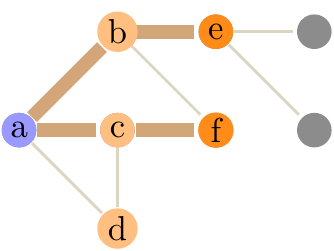}%
	}~~~
	\subfloat[]{%
	\includegraphics[clip,width=.35\columnwidth]{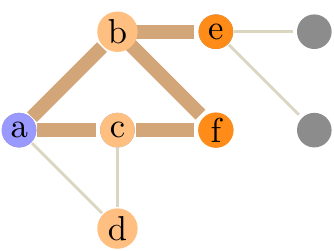}%
}
	\caption{(a) To compute the upper bound in  \eqref{eq:upperbound}, for ego $a$ and depth $k=2$, edge $\lbrace a,c\rbrace$ has the same role as edge $\lbrace a,d\rbrace$. (b) and (c) give different ways to obtain $T_2^a$. (d) shows the edges considered in shell degree.}\label{fig:S2L}
\end{figure}

As a compromise between the Ahlfors upper bound and the tree modulus lower bound, we propose a measure we call {\it shell degree}. 
Fix a depth $i=1,2,3,\dots,r$ and consider a tree rooted at the ego $a$, whose leaves are all contained in the shell $S(a,i)$, and such that the geodesics from the root to $S(a,i)$ take exactly $i$ hops. 
Let $H(a,i) = (V_i, E_i)$ be the union of all such trees found by breadth first search. For instance, in Figure \ref{fig:S2L}(d) we  show $H(\text{a},2)$ in that case. Note that we discarded nodes that are not on the geodesic paths from $a$ to $S(a,2)$.
 
Since, in general,  we cannot use the recursion \eqref{eq:recursion} on $H(a,r)$, we instead compute the upper bound \eqref{eq:upperbound}. Namely, we consider the shell connecting sets $E_i(a,k)$ for $H(a,i)$ and define the generalized shell degree to be the following expression:
\begin{equation}\label{eq:GenDeg}
\begin{split}
\text{gDeg}(a) := \sum_{i=1}^{r}\frac{1}{\sum_{k=1}^{i}\frac{1}{|E_i(a,k)|}}
\end{split}
\end{equation}
Observe that the first summand of \eqref{eq:GenDeg} is the ordinary degree of the ego and thus our formula acts as a generalization of degree which takes into account information about the shells around the ego. For example,  we have $E_1(a,1)=3$, $E_2(a,1)=2$, $E_2(a,2)=3$ in Figure \ref{fig:S2L}(d). For $r=2$, $\text{gDeg}(a) = 3+\frac{1}{\frac{1}{2}+\frac{1}{3}} = 3+6/5 = 4.2$.


We illustrate the differences between \eqref{eq:GenDeg} with \eqref{eq_ShellCFC}, \eqref{eq:upperbound}, and \eqref{eq:recursion} in Table \ref{tab:examples}  for the egonetwork in Figure \ref{fig:S2L}.

\begin{table*}
	\centering
	\caption{Examples for Shell modulus, bounds and shell degree.}
	\label{tab:examples}
	\begin{tabular}{|l|c| c|c|c|}
		\hline
		Quantity & $k=1$ & $k=2$ & $k=3$& total\\
		\hline
		\multirow{ 2}{*}{$\Mod(a,S_k)$} & \includegraphics{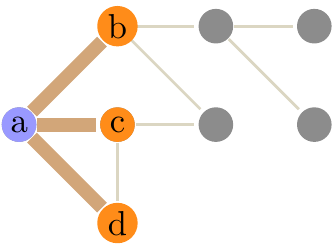} & \includegraphics{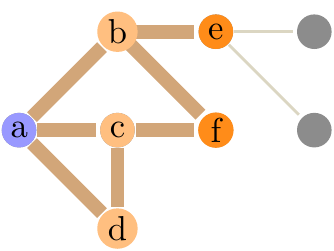}  & \includegraphics{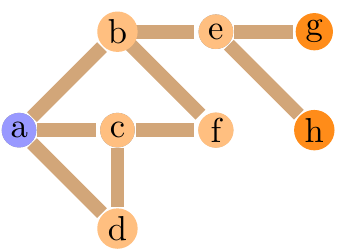}& \\
		& $3$& $1.26$ & $0.44$ & $4.71$ \\
		\hline
		\multirow{ 2}{*}{Lowerbound} & \includegraphics{S1.pdf} & \includegraphics{S2L.pdf} &  \includegraphics{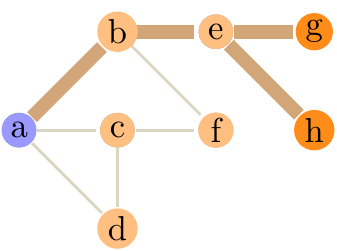} & \\
		& $3$ & $0.66$ & $0.4$ & $4.06$ \\
		\hline
		\multirow{ 2}{*}{Upperbpund} & \includegraphics{S1.pdf} & \includegraphics{S2U.pdf} & \includegraphics{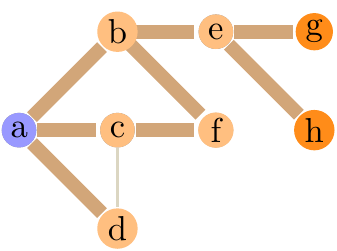} & \\
		& $3$ & $1.5$ & $0.85$ & $5.35$ \\
		\hline
		\multirow{ 2}{*}{Shell Degree} & \includegraphics{S1.pdf} & \includegraphics{S2g.pdf} & \includegraphics{S3L.pdf} & \\
		& $3$ & $1.2$ & $0.4$ & $4.6$ \\
		\hline
	\end{tabular}
\end{table*}


We can compute the summands in \eqref{eq:GenDeg} with  Algorithm \ref{alg:gdeg}. Normalization is unnecessary for shell degree, as in the case of degree,  which is critical when comparing centrality of different egos and there is no information about connections between their ego-networks.
\begin{figure}
	\begin{minipage}{\linewidth}
		\begin{algorithm}[H]
			\caption{Algorithm for computing  summands in \eqref{eq:GenDeg}.}
			\label{alg:gdeg}
			\begin{algorithmic}[1]
				\State  $D\leftarrow$ set of all descendants for each ancestor
				\State $r\leftarrow $ neighborhood order
				\State $k\leftarrow 1$
				\For{\textit{nodes} in $\lbrace S^r(a,k) ,k\le r\rbrace$}
				\State Update $D$ with \textit{nodes} as new descendants
				\State Removing ancestors that do not have any descendants in \textit{nodes}
				\State $k\leftarrow k+1$
				\EndFor				
				\State \textbf{return} harmonic means of number of ancestral relations in each $k$
			\end{algorithmic}
		\end{algorithm}
	\end{minipage}
\end{figure}

In short, we keep track of ancestral relations from the ego to nodes in each shells, and discard nodes that do not have any descendants in shell $r$; leading to required information about $H(a,r)$ and thus we can find summands in \eqref{eq:GenDeg}. The overall time complexity of calculating \eqref{eq:GenDeg} depends on the graph search in step 4 of Algorithm \ref{alg:gdeg} and keeping the information of ancestral relationships, i.e, for an ego  network $G^a(r)$  size $n_a$, algorithm performance is in $\mathcal{O}(rn_a)$.


We illustrate the performance of shell degree compared to the Ahlfors upper bound and the Tree modulus lower bound for conventional random network models such as 
Erd\H{o}s-R\'enyi networks, scale-free (Barabasi-Albert model \cite{barabasi1999emergence}), Spatial (geometric model in the unit square \cite{penrose2003random}), and small world (Watts-Strogatz model \cite{watts1998collective}). Figure \ref{fig:GenDegComparison} shows that shell degree gives a better approximation for $\mathcal{C}_{\text{shell}}(a,r)$ than the Ahlfors and Tree modulus estimates. 
\begin{figure*}
	\subfloat[]{%
		\includegraphics[clip,width=.8\columnwidth]{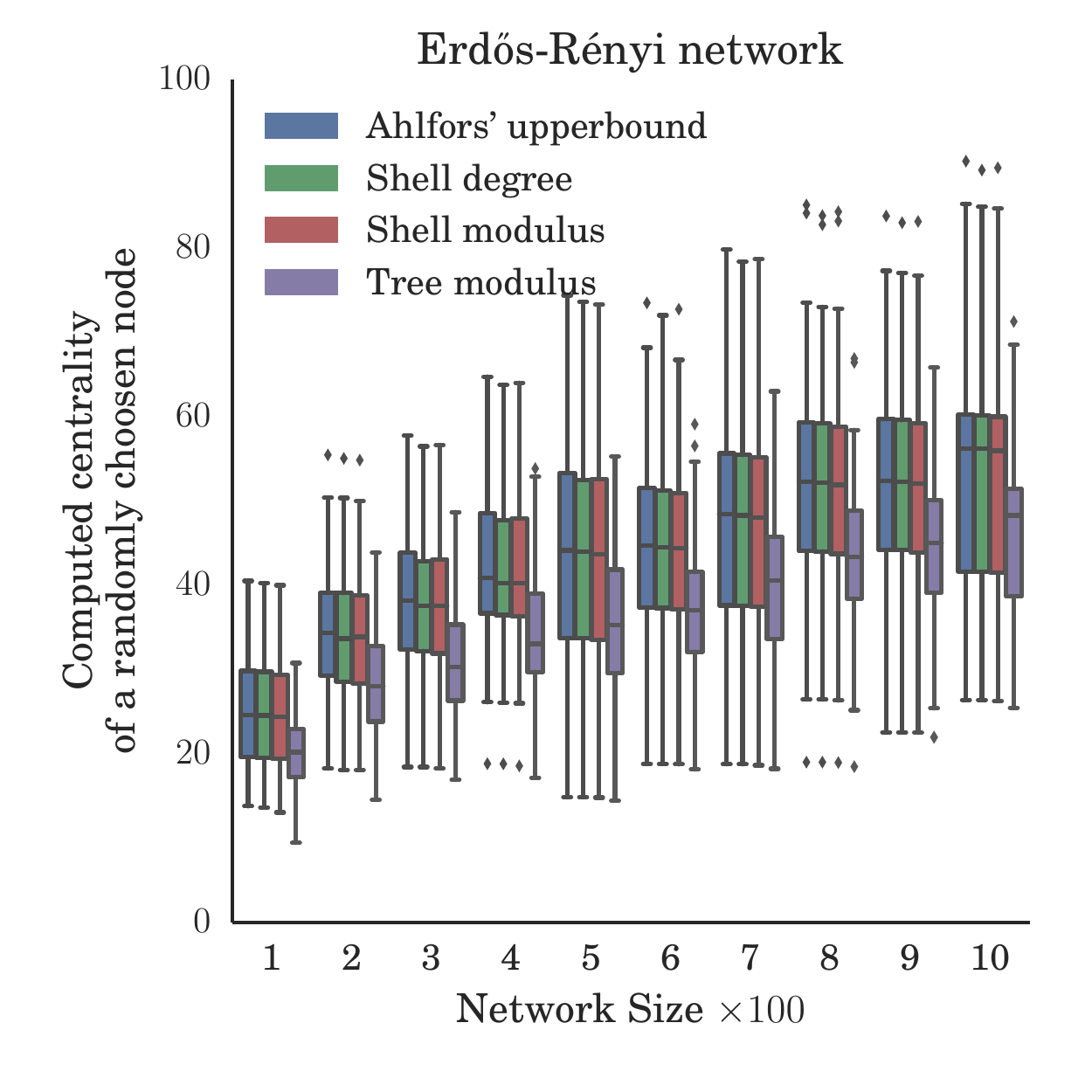}%
	}
	\subfloat[]{%
		\includegraphics[clip,width=.8\columnwidth]{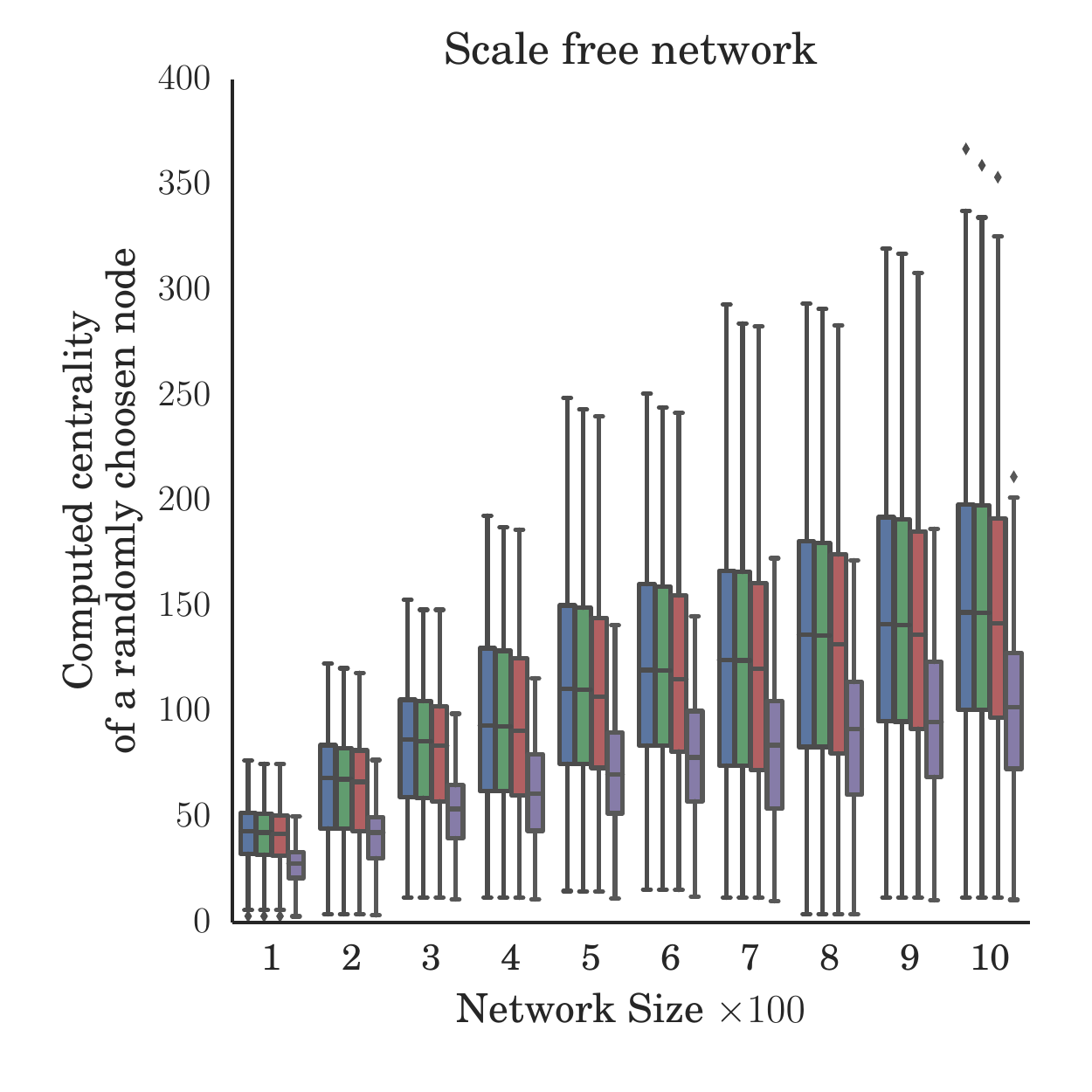}%
	}\\
	\subfloat[]{%
		\includegraphics[clip,width=.8\columnwidth]{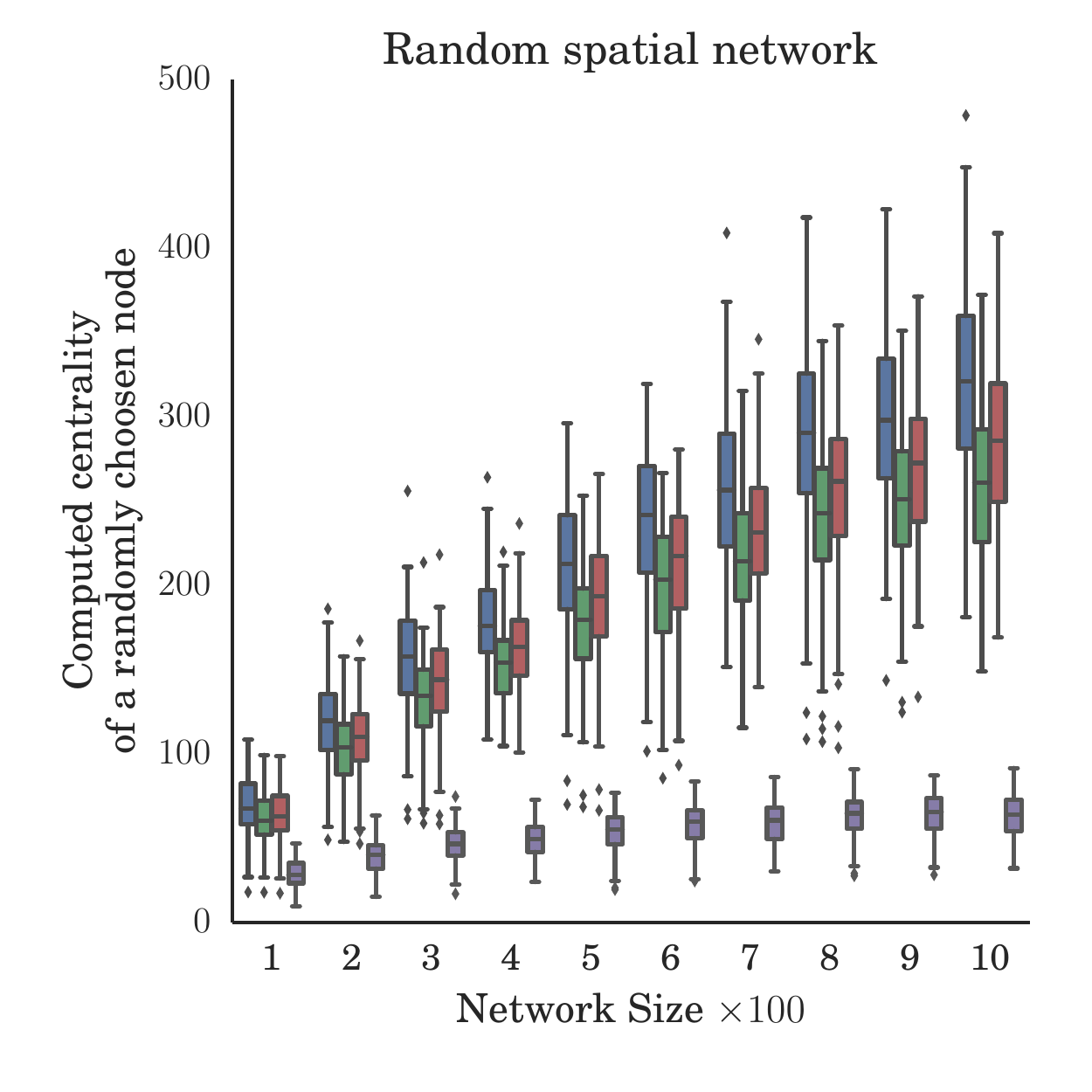}%
	}
	\subfloat[]{%
		\includegraphics[clip,width=.8\columnwidth]{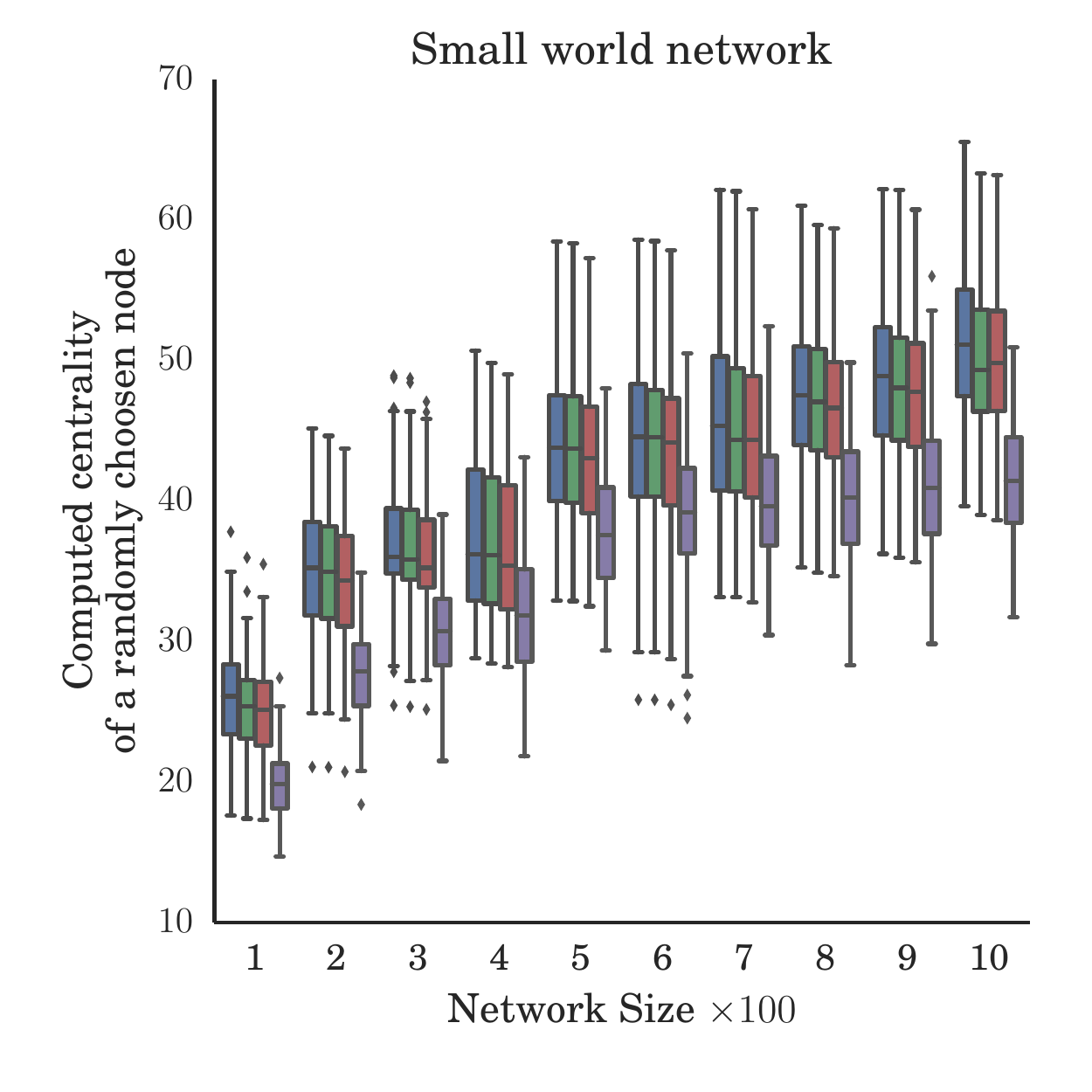}%
	}
	\caption{Comparing the value of the Ahlfors upper bound, Tree modulus lower bound, Shell degree, and Shell modulus in simulated random network models (a) Erd\H{o}s-R\'enyi networks with $p=2\log n/n$, (b) Scale free network by Barabasi and Albert model \cite{barabasi1999emergence} with $6$ edges preferential attachment. (c) Spatial network (random geometric network \cite{penrose2003random}) with distance threshold value $r= \sqrt{2\log n/n}$ and small world network by Watts-Strogatz model with initial degree of $2\log n$ and rewiring probability $0.3$. Shell degree is providing a fair estimate of shell modulus in these networks.}\label{fig:GenDegComparison}
\end{figure*}

We see that for egocentric network data with medium sizes and order of neighborhood, shell degree performs extremely well. However, it is possible to produce pathological network examples for which all of the estimates for shell modulus get worse as  $n,r\rightarrow \infty$, see Appendix \ref{ap_counterEx} for more details. 

\section{Applications of shell degree for targeted immunization strategies}\label{ap_vacc}
Targeted immunizations in computer networks and human populations can greatly impact the overall outcome of spreading processes \cite{pastor2002immunization, motter2002cascade, zhao2005enhanced}. Mitigating an epidemic with random immunization of nodes, requires vaccinating over $80\%$ of the population and thus identifying a good set of target nodes has attracted much attention \cite{chen2008finding,salathe2010dynamics}.

Most of the methods for finding good sets of nodes to immunize require global knowledge of the network, making them impossible to use in some practical situations. Therefore, scientists prefer algorithms that are agnostic relative to the global structure of the network. For example, \textit{acquaintance immunization} chooses random neighbors of randomly picked nodes \cite{cohen2003efficient}. 
In what follows, we illustrate the immunization performance of the approximation of the egocentric version of effective conductance that we call shell degree. We assume $r=3$, \textit{i.e.}, knowledge of neighbors together with neighbors of neighbors are available. The efficacy of immunization is compared to the popular egocentric measure of acquaintance centrality, and to sociocentric indices such as effective conductance, and betweenness and eigenvector centrality.

We consider the epidemic model ``susceptible, infected, recovered'' (SIR) that represents infectious processes that are not reversible. Susceptible nodes (S) in the network become infected (I) proportionally to the infectious rate $\beta$ and the number of infected neighbors, and 
eventually they rest in state (R) after a recovery period of $\frac{1}{\delta}$ days on average (see Figure \ref{fig:SIR}). We assume a constant $\delta=0.1$, \textit{i.e.}, nodes stay in state (I) an average of 10 days. To model widespread diseases such as the flu that are caused by close contacts, the infectious rate $\beta$ is chosen to have reproduction number $R_0 \sim \frac{\beta}{\delta}\langle k \rangle = 3 $, where $\langle k \rangle$ is the  average degree of the network \cite{salathe2010dynamics}.
\begin{figure}
	\includegraphics[clip,width=.9\columnwidth]{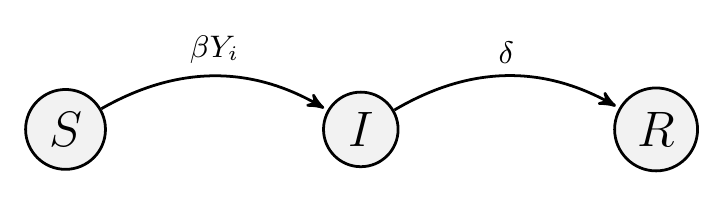}%
	\caption{Schematic of the transition graph of node $i$ in SIR moel. The infection and recovery rates are denoted by $\beta$ and $\delta$ and $Y_i$ is the number of neighbors in the infected state $I$. 
	}\label{fig:SIR}
\end{figure}
After updating the contact networks with the immunized nodes, we assess the performance of each strategy.
In our experiments, all nodes are initially susceptible and the infectious process starts from a randomly chosen patient zero. The performance of immunization strategies are  monitored by measuring the epidemic final size, \textit{i.e.}, number of nodes in  state (R) after there is no more (I) nodes.

We simulate the process $2000$ times for each immunization strategy and each immunization coverage. The simulations are done with GEMFsim, that employs event-based exact stochastic simulation \cite{sahneh2016gemfsim} for US power grid and PGP networks, and the friendship network for Princeton University extracted from Facebook \cite{traud2012social}. Salathe \textit{et. al.} \cite{salathe2010dynamics} suggest considering interactions of individuals in the same dormitory or same year and major, for the Facebook friendship networks, to capture potential physical networks--this makes the networks extremely modular.

In Figure \ref{fig_vacc_main}, each bar shows the difference of number of cases in the outbreak immunized with the two strategies shown on y-axis for a network. Positive difference (shown in red) means the alternate strategy performs better than shell degree and negative difference (shown in blue) means the shell degree gives better immunization and prevents more cases. We test the significance of comparisons of the obtained results using the nonparametric Mann-Whitney test with $\alpha = 0.05$ \cite{mann1947test} and statistically non-significant conclusions are shown by shaded colors. 

Effective conductance and betweenness centralitities perform better than shell degree for small immunization coverages. However, using sociocentric centrality measures to design targeted immunization strategies can overlook an important issue, namely, that after removing a fraction of the nodes in the network, the initial ranking by these measures is no longer valid. On the other hand, this is not as dramatic for egocentric measures such as degree, acquaintance, and the egocentric version of effective conductance and the resulting ranking is more robust after changes in the network \cite{costenbader2003stability}. 
Sociocentric measures generally struggle with this fact and thus searching for a good egocentric measure is critical. Therefore, with increasing immunization coverage, shell degree performs better (or similarly) compared to other methods. For strongly modular networks, e.g. the Princeton friendship network, closeness centrality measures are generally less efficient compared to betweenness centrality measures. However, in this case, shell degree is performing better than both eigenvector centrality and acquaintance immunization.

\begin{figure*}
	\centering
	\includegraphics[clip,width=2\columnwidth]{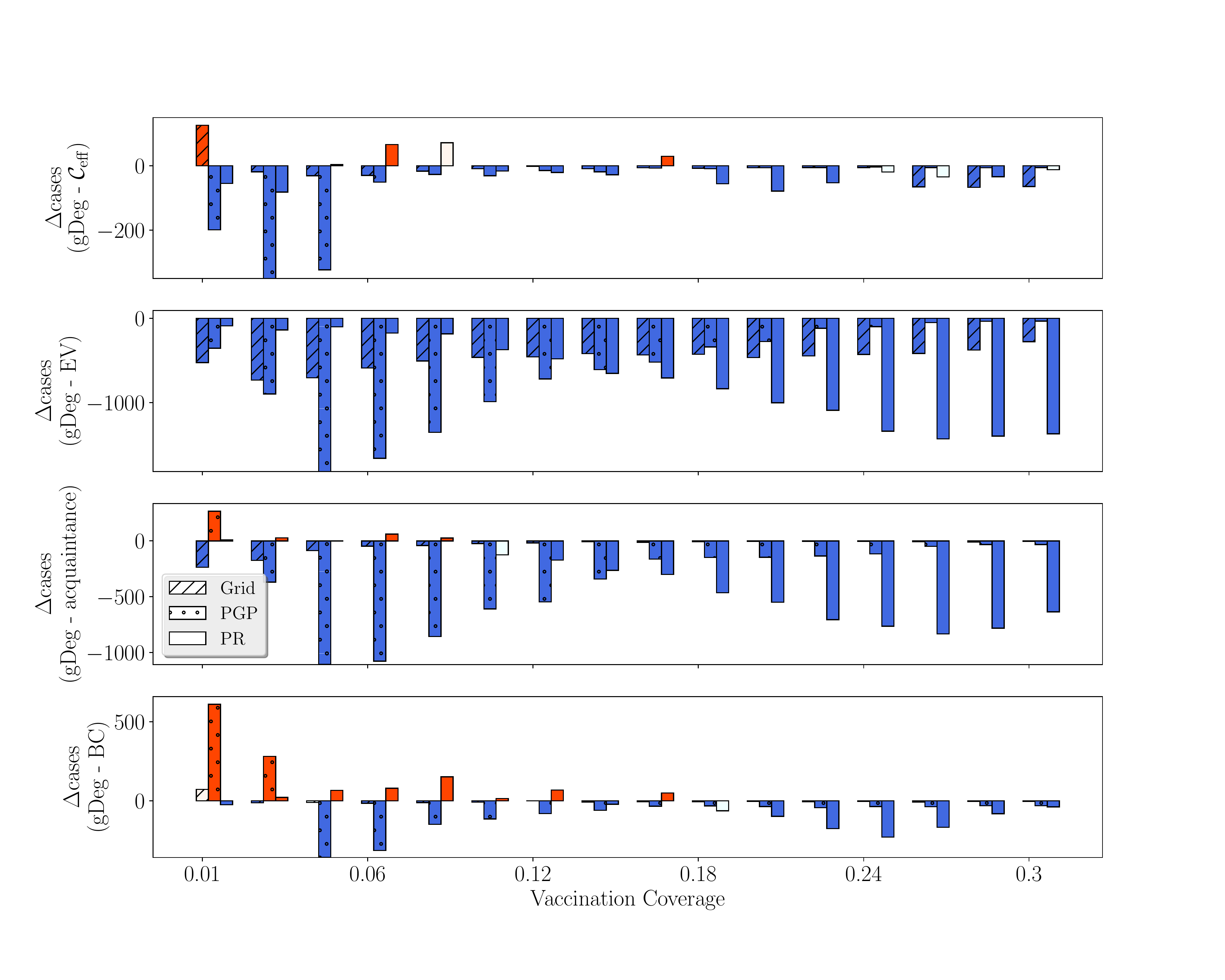}%
	\caption{Comparing different immunization strategies with effective conductance, acquaintance, eigenvector centrality, and betweenness centrality with the approximation of the egocentric version of effective conductance, that we call shell degree, with ($r=3$). The immunization coverage varies from $1\%$ to $30\%$ of the highest central nodes. Bars show the difference of the final size of the epidemic outbreak. Negative differences show that shell degree prevents more cases compared to the other policy. By increasing the coverage, shell degree outperforms other methods, since it is more robust to changes in the network structure. Results are inferred using $2000$ simulations of the SIR epidemic model and statistically nonsignificant results are shown by shaded bars.
		The empirical networks we consider are the US power grid (Grid) \cite{watts1998collective}, the PGP network (PGP) \cite{boguna2004models},  and a Facebook friendship network from Princeton university (PR) \cite{traud2012social}. }\label{fig_vacc_main}
\end{figure*}

\section{Conclusions}
In summary, we studied effective conductance centrality in the language of modulus of families of walks and in the context of egocentric networks. We compared our method to its  well-known sociocentric counterpart and illustrated the advantages of our approach. For undirected networks, shell modulus can be computed by solving a Laplacian system similar to \cite{ellens2011effective}. Moreover, for directed multi-edge networks, we propose approximations that carry the same benefits of the original definition while being easier to compute and scalable. Finally, we introduced a generalization of degree called shell degree. Applications of these tools  illustrate the advantages of the proposed measures, for instance to guide epidemic mitigation strategies under limited knowledge of the overall network.
\section*{Acknowledgments}
Authors are thankful to NSF grants DMS-1515810 and CIF-1423411.

\bibliographystyle{apsrev4-1}
\bibliography{refs}

\onecolumngrid
\appendix
\onecolumngrid
\section{Ahlfors upper bound for Erd\H{o}s-R\'enyi networks}
We want to estimate the expected Ahlfors upper bound in Erd\H{o}s-R\'enyi in the {\it connected} regime:
\[
p(N-1)=2\log N.
\] 
We can use the concavity property of Ahlfors bound and get
\[
\mathbb{E}\left( \sum_{i=1}^r\frac{1}{\sum_{k=1}^i \frac{1}{\theta_j}}\right)\le \sum_{i=1}^r\frac{1}{\sum_{k=1}^i \frac{1}{\mathbb{E}\theta_k}}
\]
we would like to estimate $\mathbb{E}(\theta_k)$.
\begin{itemize}
	\item First, note that $\theta_1$ is ${\rm Binomial}(N-1,p)$. So:
	\[
	\mathbb{E} \theta_1 =p(N-1),
	\]
	from the binomial distribution.
	\item Now, given $\theta_1$ we must toss $\theta_1$ variables distribute as ${\rm Binomial}(N-1-\theta_1,p)$, because the ego and the first shell are  now out of consideration. So
	\[
	\mathbb{E}\left( \theta_2\mid\theta_1\right)=\theta_1 p(N-1-\theta_1).
	\]
	Therefore, computing the second moment of $\theta_1$ we get:
	\[
	\mathbb{E} \theta_2 =\mathbb{E}(\mathbb{E}(\theta_2\mid\theta_1))=\mathbb{E}(\theta_1) p(N-1)-p\mathbb{E}(\theta_1^2)=p^2(1-p)(N-1)(N-2).
	\]
	\item Given $\theta_1$ and $\theta_2$ we must toss  a certain number $s$ of ${\rm Binomial}(N-1-\theta_1,p)$ random variables, where $s$ is the number of nodes in the second shell. However, this number $s$ is not easy to calculate because it depends on the interaction at the previous step. For instance, if all the binomial variables in the previous step are equal to zero, then $s=0$. But for higher values of $s$ it becomes quite complicated.
\end{itemize}

In particular, we will have
\[
\mathbb{E} \theta_1 = \log N\qquad\text{and}\qquad \mathbb{E} \theta_2\simeq (\log N)^2.
\]

\subsection{Lower bound for $\mathbb{E}(\theta_k)$}
First, we will estimate $\mathbb{E}\theta_k$ from below.
Given an ego $a$, Spielman~\cite{spielman2009spectral}  sets
\[
r(a):=\max\left\{r: |B(r,a)|\le \frac{N}{12\log N}\right\}
\]
and then shows that for $k\le r(a)$,
\[
\mathbb{P}\left[|S(a,k+1)|\le \frac{1}{5}\log N |S(a,k)|   \right]\le N^{-1.2|S(a,k)|}.
\]
He first finds that
\begin{equation}\label{eq:spielman}
\mathbb{E}\left[ |S(a,k+1)| \mid G^a(k)  \right] \ge \frac{5}{3}|S(a,k)|\log N,
\end{equation}
and then applies the theory of Chernoff bounds.
Note that by simply taking the expectation in (\ref{eq:spielman}) we get
\[
\mathbb{E}|S(a,k+1)|\ge \frac{5}{3}(\log N) \mathbb{E}|S(a,k)|.
\]
This gives  geometric growth for $k\le r(a)$:
\begin{equation}\label{eq:growth-exp}
\mathbb{E}|S(a,k)|\geq (\log N)^k.
\end{equation}
In our case, since every $c\not\in B(a,k)$ must toss $|S(a,k)|$ biased coins, we get
\[
\mathbb{E}\left[ \theta_{k+1}\mid G^a(k)\right] = |S(a,k)|p(N-|B(a,k)|)\ge \frac{11}{12}|S(a,k)|pN =\frac{11}{6}(\log N)|S(a,k)|.
\]
Again, we can take expectations and get
\[
\mathbb{E} \theta_{k+1} \ge \frac{11}{6}(\log N)\mathbb{E}|S(a,k)|.
\]
Using (\ref{eq:growth-exp}), we get
\[
\mathbb{E}\theta_{k}\ge (\log N)^{k}.
\]

\subsection{Upper bound for $\mathbb{E}\theta_k$}
To get an upper bound we can compare the growth in the Erd\H{o}s-R\'enyi graph with the growth for a Galton-Watson branching process with offspring distribution $X={\rm Binomial}(N-1,p)$. This will be larger because there are no collisions and we always toss the maximum number of coins.
If $Z_k$ is the population at time $k$, then
\[
\mathbb{E} Z_k =\mu^k
\]
where $\mu=\mathbb{E} X=p(N-1)=2\log(N)$ and we get that
\[
\mathbb{E} \theta_k\le (2\log N)^k.
\]

\subsubsection{Upper bound for the Ahlfors estimate}
We can apply this to our estimate of the average Ahlfors upper bound and get that:
\begin{align*}
\mathbb{E}\left( \sum_{k=1}^r\frac{1}{\sum_{j=1}^k \frac{1}{\theta_j}}\right)
& \le \sum_{k=1}^r\frac{1}{\sum_{j=1}^k \frac{1}{\mathbb{E}\theta_j}}\\
& \le \sum_{k=1}^r\frac{1}{\sum_{j=1}^k \frac{1}{(2\log N)^j}}\\
& = \left(2\log N-1\right)\sum_{k=1}^r \left[1+\frac{1}{(2\log N)^{k}-1}\right]\\
& \simeq \left(2\log N-1\right)\left[ r+\sum_{k=1}^r\frac{1}{(2\log N)^k}\right]\\
& =\left(2\log N-1\right)\left[r+\frac{1}{2\log N}\frac{1-\left(\frac{1}{2\log N}\right)^r}{1-\frac{1}{2\log N}}\right]\\
& = \left(2\log N-1\right)\left[r+1-\frac{1}{(2\log N)^r((2\log N)-1)}\right]\\
& \simeq (2\log N-1)(r+1)
\end{align*}

\section{Behavior of shell modulus estimates when $n,r \rightarrow \infty$}\label{ap_counterEx}
\subsection{Modulus on the complete graph}

Verifying that a metric $\rho$ is extremal for $p$-modulus can be done using Beurling's criterion (proof in \cite{albin2016minimal}).  
\begin{theorem}[Beurling's Criterion for Extremality]
	\label{thm:beurling}
	Let $G$ be a simple graph, $\Gamma$ a family of walks on $G$, and $1<p<\infty$.  
	Then, a density $\rho\in \Adm(\Gamma)$ is extremal for $\Mod_p(\Gamma)$, if there is a subfamily $\tilde{\Gamma}\subset\Gamma$ with $\ell_\rho(\gamma)=1$ for all $\gamma\in \tilde{\Gamma}$, such that for all $h\in\mathbb{R}^E$:
	\begin{equation}
	\label{eq:beurling}
	\mbox{$\sum_{e\in E}\mathcal{N}(\gamma,e)h(e)\geq 0$,\quad for all $\gamma\in\tilde{\Gamma}$}\quad\Longrightarrow\quad\sum_{e\in E}h(e)\rho^{p-1}(e)\geq 0.
	\end{equation}
\end{theorem}

The {\it complete graph} $K_N$ is a simple graph on $N$ nodes, where every node is connected to each other, see Figure \ref{fig:complete}.

\begin{figure}[h!]
	\centering
	\includegraphics[trim={0cm 5cm 0cm 5cm},clip,width= 0.3\linewidth]{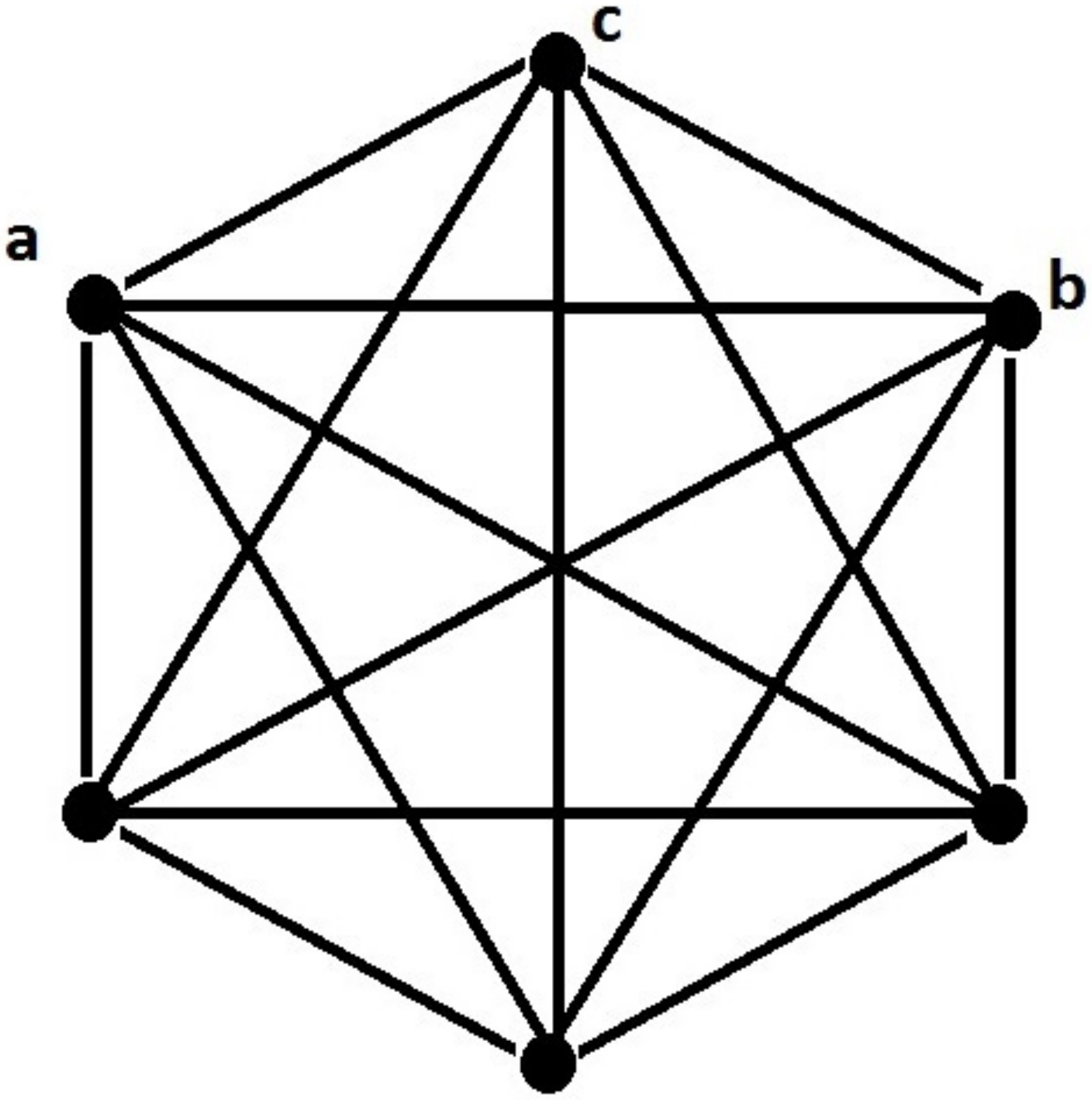}
	\caption{$K_6$- Complete graph on 6 nodes}\label{fig:complete}
\end{figure}

Figure \ref{fig:rho-complete} depicts the extremal density $\rho^*$ for  $\Gamma(a,b)$ in  $K_N$.

\begin{figure}[h!]
	\centering
	\includegraphics[trim={0cm 8cm 0cm 8cm},clip,width= 0.4\linewidth]{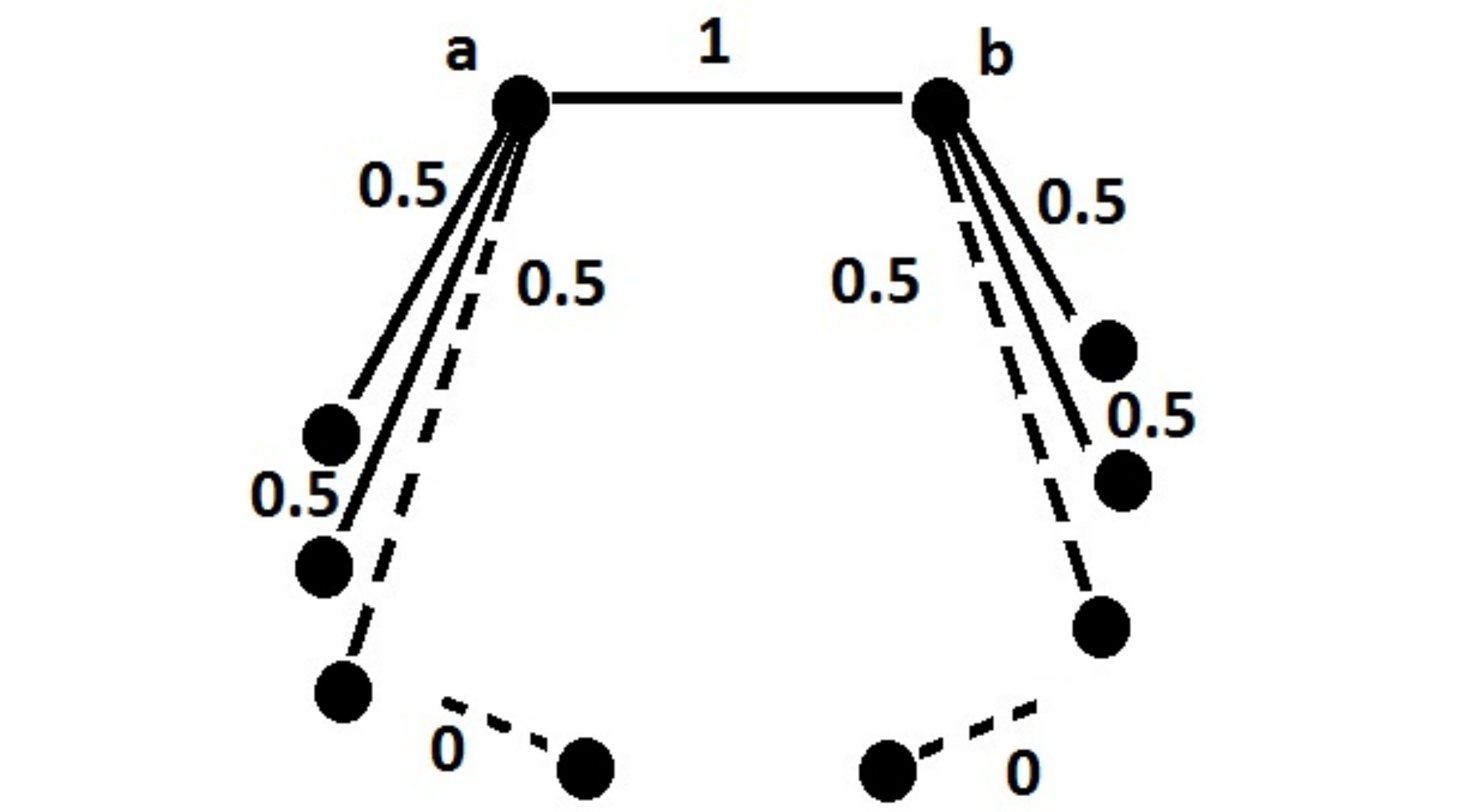}
	\caption{$\rho^*$ for $\Gamma(a,b)$ on $K_N$}\label{fig:rho-complete}
\end{figure}

In formulas, $\rho^*(a,x)=1/2=\rho^*(b,x)$ for every $x\neq a,b$, and $\rho^*(a,b)=1$, otherwise $\rho^*$ is zero. To verify Beurling's criterion, consider the subfamily $\tilde{\Gamma}$ of simple paths consisting of $a\ b$ and $a\ x\ b$ for any $x\neq a,b$.
We get that
\[
\Mod_p(\Gamma(a,b)) = 1+ 2 (N-2) \frac{1}{2^p}\qquad\text{and}\qquad\Mod_2(\Gamma(a,b)=\frac{N}{2}.
\]

Take $n$ complete graphs $K_1,\dots , K_n$.

\subsection{Modulus on a chain of complete graphs}
\paragraph{Constant sizes}

For $j=1,\dots,n$, assume that $|V(K_j)|=N$ , and pick a pair of distinct nodes $x_{j-1},y_j\in V(K_j)$. Then, for $j=1,\dots,n-1$, glue $y_j\in V(K_j)$ to $x_j\in V(K_{j+1})$.
We denote the resulting graph by $G(N,n)$.

For convenience, for $j=1,\dots,n$, we write $A_j:=V(K_j)\setminus\{x_{j-1},y_j\}$, so that the shell at level $j$ is $S_j=V(K_j)\setminus\{x_{j-1}\}=A_j\cup\{y_j\}$.
Then, fix $m=1,\dots,n$, and for $j=1,\dots,m-1$, define the following density on $\in E(K_j)$:
{\renewcommand\arraystretch{1.8} 				
	\[
	\rho^*(e):=\left\{ 
	\begin{array}{ll}
	\frac{1}{m} & \text{if $e=\{x_{j-1},y_j\}$}\\
	\frac{1}{2m} & \text{if $e=\{x_{j-1},a\}$ or $e=\{y_j,a\}$ for some $a\in A_j$}\\
	0 & \text{otherwise}
	\end{array}
	\right.
	\]
}
For $j=m$, and $e\in E(K_m)$,  set
{\renewcommand\arraystretch{1.8} 				
	\[
	\rho^*(e):=\left\{ 
	\begin{array}{ll}
	\frac{1}{m} & \text{if $e=\{x_{m-1},a\}$ for some $a\in A_m\cup\{y_m\}$}\\
	0 & \text{otherwise}
	\end{array}
	\right.
	\]
}
Observe that the support of $\rho^*$ can be decomposed as the disjoint union of $N-1$ paths. To see this, enumerate each $A_j=\{a_{j,k}\}_{k=1}^{N-2}$. Then, for $k=1,\dots,N-2$, let 
\[
\gamma_{m,k} := x_0\ a_{1,k}\ x_1\ a_{2,k}\ \cdots\ x_{m-1}\ a_{m,k}.
\]
Finally set
\[
\gamma_{m,0} := x_0\ y_1\ \cdots\ x_{m-1}\ y_m.
\]
One can check that $\tilde{\Gamma}=\{\gamma_{m,k}\}_{k=0}^{N-2}$ is a Beurling subfamily for the shell modulus $\Mod_2(x_0,S_m)$. So
\[
\Mod_2(x_0,S_m)=\frac{1}{m}+(N-2)\left[\frac{2m-2}{4m^2}+\frac{1}{m^2}\right]=\frac{N}{2m}\left(1+\frac{1}{m}\right)-\frac{1}{m^2},
\]
which is roughly $N/(2m)$. Also note that for $m=1$ we recover the degree of $x_0$. If we sum we get
\[
\sum_{m=1}^n \Mod_2(x_0,S_m)\simeq\frac{N}{2}\sum_{m=1}^n \frac{1}{m}\simeq \frac{N}{2}\log n.
\]
The Ahlfors upper bound gives
\[
\sum_{m=1}^n\frac{1}{\sum_{j=1}^m\frac{1}{N-1}}=(N-1)\sum_{m=1}^n\frac{1}{m}\simeq (N-1)\log n.
\]
The generalized shell degree, gives
\[
\sum_{m=1}^n\frac{1}{m-1+\frac{1}{N-1}}\simeq N + \log n
\]

\paragraph{Increasing sizes}

Now we repeat the construction above, but this time, setting $k_j:=|V(K_j)|$, we have $k_1=\alpha_1+2$ and,
for $j=2,\dots,n$, we assume that $k_j=\alpha_j (k_{j-1}-2)+2$, for an increasing sequence of positive integers $\{\alpha_j\}_{j=2}^n$.

Then, fix $m=1,\dots,n$, and for $j=1,\dots,m-1$, define the following density on $\in E(K_j)$:
{\renewcommand\arraystretch{1.8} 				
	\[
	\rho^*(e):=\left\{ 
	\begin{array}{ll}
	\frac{\prod_{k=j+1}^m\alpha_{k}}{1+\sum_{j=1}^{m}\prod_{k=j+1}^m\alpha_{k}} & \text{if $e=\{x_{j-1},y_j\}$}\\
	\frac{2^{-1}\prod_{k=j+1}^m\alpha_{k}}{1+\sum_{j=1}^{m}\prod_{k=j+1}^m\alpha_{k}} & \text{if $e=\{x_{j-1},a\}$ or $e=\{y_j,a\}$ for some $a\in A_j$}\\
	0 & \text{otherwise}
	\end{array}
	\right.
	\]
}
For $j=m$, and $e\in E(K_m)$,  set
{\renewcommand\arraystretch{1.8} 				
	\[
	\rho^*(e):=\left\{ 
	\begin{array}{ll}
	\frac{1}{1+\sum_{j=1}^{m}\prod_{k=j+1}^m\alpha_{k}} & \text{if $e=\{x_{m-1},a\}$ for some $a\in A_m\cup\{y_m\}$}\\
	0 & \text{otherwise}
	\end{array}
	\right.
	\]
}
Now form $k_m-1$ paths.
Set
\[
\gamma_{m,0} := x_0\ y_1\ \cdots\ x_{m-1}\ y_m.
\]
As before, enumerate each $A_j=\{a_{j,k}\}_{k=1}^{k_j-2}$.
Now, $k_m-2=\alpha_m(k_{m-1}-2)$, so we can group the $k_m-2$ edges $\{x_{m-1},a\}$ for $a\in A_m$ into $k_{m-1}-2$ groups of $\alpha_m$ edges. Each such group will then flow through a different node in $A_{m-1}$, and then we repeat. The claim is that this gives rise to a Beurling family of paths $\tilde{\Gamma}$. By construction, they all have $\rho^*$ length equal to $1$. We only need to check Beurling's criterion. So suppose $h\in \mathbb{R}^E$ satisfies
\[
\ell_h(\gamma)\ge 0\qquad\text{for all $\gamma\in\tilde{\Gamma}$}.
\]
Then $\sum_{e\in E}\rho^*(e)h(e)$ is equal to:
\[
\sum_{j=1}^m (\rho^* h)(x_{j-1},y_j)+\sum_{j=1}^{m-1}\sum_{i=1}^{k_j-2}[(\rho^* h)(x_{j-1},a_{j,k})+(\rho^* h)(a_{j,k},y_j)]+\sum_{i=1}^{k_m-2} (\rho^* h)(x_{m-1},a_{m,k}).
\]
And if we write $\alpha:=1+\sum_{j=1}^{m}\prod_{k=j+1}^m\alpha_{k}$, and collect terms, this equals
\[
\alpha^{-1}\left(
\alpha\sum_{j=1}^m h(x_{j-1},y_j)+\sum_{j=1}^{m-1}\left(\prod_{k=j+1}^m\alpha_{k}\right)\sum_{i=1}^{k_j-2}[h(x_{j-1},a_{j,k})+h(a_{j,k},y_j)]+\sum_{i=1}^{k_m-2} h(x_{m-1},a_{m,k})\right).
\]
which is $\ge 0$, because for every $j=1.\dots,m-1$
\[
(k_j-2)\prod_{k=j+1}^m\alpha_{k}=k_m-2
\]
So we get
\[
\Mod_2(x_0,S_m)=\alpha^{-2}\left(1+\frac{3}{2}(k_m-2)\sum_{j=1}^{m}\prod_{k=j+1}^m\alpha_{k} +(k_m-2) \right)
\]

Now choose $\alpha_j\equiv 2$. Then
\[
\alpha=1+2+4+\cdots+2^{m-1}=2^m -1.
\]
Also
\[
k_m-2=2^{m-1}\alpha_1
\]
So
\[
\Mod_2(x_0,S_m)\simeq  \alpha_1.
\]
And
\[
\sum_{m=1}^n\Mod_2(x_0,S_m)\simeq \alpha_1 n.
\]
On the other hand the shell degree is
\[
\sum_{m=1}^n \frac{1}{m-1+\frac{1}{k_m-1}}\simeq\log n.
\]

\end{document}